\documentclass[9pt, twocolumn]{IEEEtran}%
\usepackage[T1]{fontenc}
\usepackage[utf8]{inputenc}
\usepackage{mathtools}

\usepackage{amssymb,mathrsfs}% Typical maths resource packages
\usepackage{amsthm}
\usepackage{bm}
\usepackage{scalerel}
\usepackage{nicefrac}
\usepackage{microtype} 
\usepackage[shortlabels]{enumitem}
\usepackage{graphicx}
\usepackage{epstopdf}
\DeclareGraphicsExtensions{.eps,.png,.jpg,.pdf}

\usepackage{url}
\usepackage{colortbl}
\usepackage{booktabs}
\usepackage{multirow}
\usepackage{colortbl,xcolor}
\usepackage[normalem]{ulem}
\usepackage{xparse,xstring}
\usepackage{calc}
\usepackage{etoolbox}

\makeatletter
\@ifpackageloaded{natbib}{
	\relax
}{
	\usepackage{cite}
}
\makeatother

\usepackage{array}
\newcolumntype{L}[1]{>{\raggedright\let\newline\\\arraybackslash\hspace{0pt}}m{#1}}
\newcolumntype{C}[1]{>{\centering\let\newline\\\arraybackslash\hspace{0pt}}m{#1}}
\newcolumntype{R}[1]{>{\raggedleft\let\newline\\\arraybackslash\hspace{0pt}}m{#1}}

\makeatletter
\let\MYcaption\@makecaption
\makeatother
\usepackage[font=footnotesize]{subcaption}
\makeatletter
\let\@makecaption\MYcaption
\makeatother

\usepackage{glossaries}
\makeatletter
\sfcode`\.1006

\let\oldgls\gls
\let\oldglspl\glspl

\newcommand\fussy@ifnextchar[3]{%
	\let\reserved@d=#1%
	\def\reserved@a{#2}%
	\def\reserved@b{#3}%
	\futurelet\@let@token\fussy@ifnch}
\def\fussy@ifnch{%
	\ifx\@let@token\reserved@d
		\let\reserved@c\reserved@a
	\else
		\let\reserved@c\reserved@b
	\fi
	\reserved@c}

\renewcommand{\gls}[1]{%
\oldgls{#1}\fussy@ifnextchar.{\@checkperiod}{\@}}
\renewcommand{\glspl}[1]{%
\oldglspl{#1}\fussy@ifnextchar.{\@checkperiod}{\@}}

\newcommand{\@checkperiod}[1]{%
	\ifnum\sfcode`\.=\spacefactor\else#1\fi
}

\robustify{\gls}
\robustify{\glspl}
\makeatother

\newacronym{wrt}{w.r.t.}{with respect to}
\newacronym{RHS}{R.H.S.}{right-hand side}
\newacronym{LHS}{L.H.S.}{left-hand side}
\newacronym{iid}{i.i.d.}{independent and identically distributed}
\newacronym{SOTA}{SOTA}{state-of-the-art}
\usepackage{float}

\ifx\notloadhyperref\undefined
	\ifx\loadbibentry\undefined
		\usepackage[hidelinks,hypertexnames=false]{hyperref}
	\else
		\usepackage{bibentry}
		\makeatletter\let\saved@bibitem\@bibitem\makeatother
		\usepackage[hidelinks,hypertexnames=false]{hyperref}
		\makeatletter\let\@bibitem\saved@bibitem\makeatother
	\fi
\else
	\ifx\loadbibentry\undefined
		\relax
	\else
		\usepackage{bibentry}
	\fi
\fi

\usepackage[capitalize]{cleveref}
\crefname{equation}{}{}
\Crefname{equation}{}{}
\crefname{claim}{claim}{claims}
\crefname{step}{step}{steps}
\crefname{line}{line}{lines}
\crefname{condition}{condition}{conditions}
\crefname{dmath}{}{}
\crefname{dseries}{}{}
\crefname{dgroup}{}{}

\crefname{Problem}{Problem}{Problems}
\crefformat{Problem}{Problem~#2#1#3}
\crefrangeformat{Problem}{Problems~#3#1#4 to~#5#2#6}

\crefname{Theorem}{Theorem}{Theorems}
\crefname{Corollary}{Corollary}{Corollaries}
\crefname{Proposition}{Proposition}{Propositions}
\crefname{Lemma}{Lemma}{Lemmas}
\crefname{Definition}{Definition}{Definitions}
\crefname{Example}{Example}{Examples}
\crefname{Assumption}{Assumption}{Assumptions}
\crefname{Remark}{Remark}{Remarks}
\crefname{Rem}{Remark}{Remarks}
\crefname{remarks}{Remarks}{Remarks}
\crefname{Appendix}{Appendix}{Appendices}
\crefname{Supplement}{Supplement}{Supplements}
\crefname{Exercise}{Exercise}{Exercises}
\crefname{Theorem_A}{Theorem}{Theorems}
\crefname{Corollary_A}{Corollary}{Corollaries}
\crefname{Proposition_A}{Proposition}{Propositions}
\crefname{Lemma_A}{Lemma}{Lemmas}
\crefname{Definition_A}{Definition}{Definitions}

\usepackage{crossreftools}
\ifx\notloadhyperref\undefined
\else
	\relax
\fi

\usepackage{algorithm}
\usepackage{algpseudocode}

\ifx\loadbreqn\undefined
	\relax
\else
	\usepackage{breqn}
\fi

\makeatletter
\def\cleartheorem#1{%
    \expandafter\let\csname#1\endcsname\relax
    \expandafter\let\csname c@#1\endcsname\relax
}
\def\clearthms#1{ \@for\tname:=#1\do{\cleartheorem\tname} }
\makeatother

\ifx\renewtheorem\undefined
	\ifx\useTheoremCounter\undefined
		\newtheorem{Theorem}{Theorem}
		\newtheorem{Corollary}{Corollary}
		\newtheorem{Proposition}{Proposition}
		\newtheorem{Lemma}{Lemma}
	\else
		\newtheorem{Theorem}{Theorem}

	\fi

\fi

\theoremstyle{remark}

\theoremstyle{plain}

\newcommand{\qednew}{\nobreak \ifvmode \relax \else
		\ifdim\lastskip<1.5em \hskip-\lastskip
			\hskip1.5em plus0em minus0.5em \fi \nobreak
		\vrule height0.75em width0.5em depth0.25em\fi}



\NewDocumentCommand{\movedownsub}{e{^_}}{%
	\IfNoValueTF{#1}{%
		\IfNoValueF{#2}{^{}}% neither ^ nor _, do nothing; if no ^ but _, add ^{}
	}{%
		^{#1}% add superscript if present
	}%
	\IfNoValueF{#2}{_{#2}}% add subscript if present
}

\let\latexchi\chi
\RenewDocumentCommand{\chi}{}{\latexchi\movedownsub}

\newcommand{\calE}{\mathcal{E}}

\newcommand{\calV}{\mathcal{V}}

\newcommand{\bG}{\mathbf{G}}

\newcommand{\bS}{\mathbf{S}}

\newcommand{\bW}{\mathbf{W}}

\newcommand{\bbN}{\mathbb{N}}

\DeclareSymbolFont{bsfletters}{OT1}{cmss}{bx}{n}
\DeclareSymbolFont{ssfletters}{OT1}{cmss}{m}{n}
\DeclareMathSymbol{\bsfGamma}{0}{bsfletters}{'000}
\DeclareMathSymbol{\ssfGamma}{0}{ssfletters}{'000}
\DeclareMathSymbol{\bsfDelta}{0}{bsfletters}{'001}
\DeclareMathSymbol{\ssfDelta}{0}{ssfletters}{'001}
\DeclareMathSymbol{\bsfTheta}{0}{bsfletters}{'002}
\DeclareMathSymbol{\ssfTheta}{0}{ssfletters}{'002}
\DeclareMathSymbol{\bsfLambda}{0}{bsfletters}{'003}
\DeclareMathSymbol{\ssfLambda}{0}{ssfletters}{'003}
\DeclareMathSymbol{\bsfXi}{0}{bsfletters}{'004}
\DeclareMathSymbol{\ssfXi}{0}{ssfletters}{'004}
\DeclareMathSymbol{\bsfPi}{0}{bsfletters}{'005}
\DeclareMathSymbol{\ssfPi}{0}{ssfletters}{'005}
\DeclareMathSymbol{\bsfSigma}{0}{bsfletters}{'006}
\DeclareMathSymbol{\ssfSigma}{0}{ssfletters}{'006}
\DeclareMathSymbol{\bsfUpsilon}{0}{bsfletters}{'007}
\DeclareMathSymbol{\ssfUpsilon}{0}{ssfletters}{'007}
\DeclareMathSymbol{\bsfPhi}{0}{bsfletters}{'010}
\DeclareMathSymbol{\ssfPhi}{0}{ssfletters}{'010}
\DeclareMathSymbol{\bsfPsi}{0}{bsfletters}{'011}
\DeclareMathSymbol{\ssfPsi}{0}{ssfletters}{'011}
\DeclareMathSymbol{\bsfOmega}{0}{bsfletters}{'012}
\DeclareMathSymbol{\ssfOmega}{0}{ssfletters}{'012}

\makeatletter
\newcommand*\rel@kern[1]{\kern#1\dimexpr\macc@kerna}
\newcommand*\widebar[1]{%
  \begingroup
  \def\mathaccent##1##2{%
    \rel@kern{0.8}%
    \overline{\rel@kern{-0.8}\macc@nucleus\rel@kern{0.2}}%
    \rel@kern{-0.2}%
  }%
  \macc@depth\@ne
  \let\math@bgroup\@empty \let\math@egroup\macc@set@skewchar
  \mathsurround\z@ \frozen@everymath{\mathgroup\macc@group\relax}%
  \macc@set@skewchar\relax
  \let\mathaccentV\macc@nested@a
  \macc@nested@a\relax111{#1}%
  \endgroup
}
\makeatother

\DeclareMathOperator{\var}{var}

\DeclareMathOperator{\cov}{cov}

\newcommand{\ifbcdot}[1]{\ifblank{#1}{\cdot}{#1}}

\DeclarePairedDelimiterX\abs[1]{\lvert}{\rvert}{\ifbcdot{#1}}
\DeclarePairedDelimiterX\parens[1]{(}{)}{\ifbcdot{#1}}
\DeclarePairedDelimiterX\brk[1]{[}{]}{\ifbcdot{#1}}
\DeclarePairedDelimiterX\braces[1]{\{}{\}}{\ifbcdot{#1}}
\DeclarePairedDelimiterX\angles[1]{\langle}{\rangle}{\ifblank{#1}{\cdot,\cdot}{#1}}
\DeclarePairedDelimiterX\ip[2]{\langle}{\rangle}{\ifbcdot{#1},\ifbcdot{#2}}
\DeclarePairedDelimiterX\norm[1]{\lVert}{\rVert}{\ifbcdot{#1}}
\DeclarePairedDelimiterX\ceil[1]{\lceil}{\rceil}{\ifbcdot{#1}}
\DeclarePairedDelimiterX\floor[1]{\lfloor}{\rfloor}{\ifbcdot{#1}}

\DeclareFontFamily{U}{matha}{\hyphenchar\font45}
\DeclareFontShape{U}{matha}{m}{n}{
      <5> <6> <7> <8> <9> <10> gen * matha
      <10.95> matha10 <12> <14.4> <17.28> <20.74> <24.88> matha12
      }{}
\DeclareSymbolFont{matha}{U}{matha}{m}{n}
\DeclareFontSubstitution{U}{matha}{m}{n}

\DeclareFontFamily{U}{mathx}{\hyphenchar\font45}
\DeclareFontShape{U}{mathx}{m}{n}{
      <5> <6> <7> <8> <9> <10>
      <10.95> <12> <14.4> <17.28> <20.74> <24.88>
      mathx10
      }{}
\DeclareSymbolFont{mathx}{U}{mathx}{m}{n}
\DeclareFontSubstitution{U}{mathx}{m}{n}

\DeclareMathDelimiter{\vvvert}{0}{matha}{"7E}{mathx}{"17}
\DeclarePairedDelimiterX\vertiii[1]{\vvvert}{\vvvert}{\ifbcdot{#1}}

\DeclarePairedDelimiterXPP\trace[1]{\operatorname{Tr}}{(}{)}{}{\ifbcdot{#1}} % column vector
\DeclarePairedDelimiterXPP\col[1]{\operatorname{col}}{\{}{\}}{}{\ifbcdot{#1}} % column vector
\DeclarePairedDelimiterXPP\row[1]{\operatorname{row}}{\{}{\}}{}{\ifbcdot{#1}} % row vector
\DeclarePairedDelimiterXPP\erf[1]{\operatorname{erf}}{(}{)}{}{\ifbcdot{#1}}
\DeclarePairedDelimiterXPP\erfc[1]{\operatorname{erfc}}{(}{)}{}{\ifbcdot{#1}}
\DeclarePairedDelimiterXPP\KLD[2]{D}{(}{)}{}{\ifbcdot{#1}\, \delimsize\|\, \ifbcdot{#2}} % KL divergence
\DeclarePairedDelimiterXPP\op[2]{\operatorname{#1}}{(}{)}{}{#2} % general operator

\newcommand{\setcomp}{^{\mkern-1.5mu\mathsf{c}}} %set complement
\newcommand{\ud}{\,\mathrm{d}} % for integrals like \int f(x) \ud x
\DeclarePairedDelimiterXPP\indicate[1]{{\bf 1}}{\{}{\}}{}{\ifbcdot{#1}}

\NewDocumentCommand\ofrac{s m}{%
	\IfBooleanTF#1%
	{\dfrac{1}{#2}}%
	{\frac{1}{#2}}%
}
\NewDocumentCommand\ddfrac{s m m}{%
	\IfBooleanTF#1%
	{\dfrac{\mathrm{d} {#2}}{\mathrm{d} {#3}}}%
	{\frac{\mathrm{d} {#2}}{\mathrm{d} {#3}}}%
}
\NewDocumentCommand\ppfrac{s m m}{%
	\IfBooleanTF#1%
	{\dfrac{\partial {#2}}{\partial {#3}}}%
	{\frac{\partial {#2}}{\partial {#3}}}%
}

\providecommand\given{}
\DeclarePairedDelimiterX\Set[2]\{\}{%
\renewcommand\given{\SetSymbol[\delimsize]{#1}}
#2
}
\DeclarePairedDelimiterX\Setc[1]\{\}{%
\renewcommand\given{\SetSymbol{:}}
#1
}

\NewDocumentCommand\set{s o m}{%
	\IfBooleanTF#1%
	{\IfValueTF{#2}{\Set*{#2}{#3}}{\Setc*{#3}}}%
	{\IfValueTF{#2}{\Set{#2}{#3}}{\Setc{#3}}}%
}

\NewDocumentCommand{\evalat}{ s O{\big} m e{_^} }{%
\IfBooleanTF{#1}%
{\left. #3 \right|}{#3#2|}%
\IfValueT{#4}{_{#4}}%
\IfValueT{#5}{^{#5}}%
}

\providecommand\given{}
\DeclarePairedDelimiterXPP\cprob[1]{}(){}{
\renewcommand\given{\nonscript\,\delimsize\vert\allowbreak\nonscript\,\mathopen{}}%
#1%
}
\DeclarePairedDelimiterXPP\cexp[1]{}[]{}{
\renewcommand\given{\nonscript\,\delimsize\vert\allowbreak\nonscript\,\mathopen{}}%
#1%
}

\DeclareDocumentCommand \P { s e{_^} d() g } {%
	\mathbb{P}%
	\IfBooleanTF{#1}%
		{
			\IfValueT{#2}{_{#2}}%
			\IfValueT{#3}{^{#3}}%
			\IfValueTF{#5}{\cprob{#4 \given #5}}{\IfValueT{#4}{\cprob{#4}}}%
		}%
		{
			\IfValueT{#2}{_{#2}}%
			\IfValueT{#3}{^{#3}}%
			\IfValueTF{#5}{\cprob*{#4 \given #5}}{\IfValueT{#4}{\cprob*{#4}}}%
		}%
}

\DeclareDocumentCommand \E { s e{_^} o g } {%
	\mathbb{E}%
	\IfBooleanTF{#1}%
		{
			\IfValueT{#2}{_{#2}}%
			\IfValueT{#3}{^{#3}}%
			\IfValueTF{#5}{\cexp{#4 \given #5}}{\IfValueT{#4}{\cexp{#4}}}%
		}%
		{
			\IfValueT{#2}{_{#2}}%
			\IfValueT{#3}{^{#3}}%	
			\IfValueTF{#5}{\cexp*{#4 \given #5}}{\IfValueT{#4}{\cexp*{#4}}}%		
		}%
}

\DeclareDocumentCommand \Var { s e{_^} d() g } {%
	\var%
	\IfBooleanTF{#1}%
		{
			\IfValueT{#2}{_{#2}}%
			\IfValueT{#3}{^{#3}}%
			\IfValueTF{#5}{\cprob{#4 \given #5}}{\IfValueT{#4}{\cprob{#4}}}%
		}%
		{
			\IfValueT{#2}{_{#2}}%
			\IfValueT{#3}{^{#3}}%	
			\IfValueTF{#5}{\cprob*{#4 \given #5}}{\IfValueT{#4}{\cprob*{#4}}}%		
		}%
}

\DeclareDocumentCommand \Cov { s e{_^} d() g } {%
	\cov%
	\IfBooleanTF{#1}%
		{
			\IfValueT{#2}{_{#2}}%
			\IfValueT{#3}{^{#3}}%
			\IfValueTF{#5}{\cprob{#4 \given #5}}{\IfValueT{#4}{\cprob{#4}}}%
		}%
		{
			\IfValueT{#2}{_{#2}}%
			\IfValueT{#3}{^{#3}}%	
			\IfValueTF{#5}{\cprob*{#4 \given #5}}{\IfValueT{#4}{\cprob*{#4}}}%		
		}%
}

\ExplSyntaxOn
\NewDocumentCommand \dist {m o o} {%
\mathrm{#1}\left(%
	\IfValueT{#3}{%
		\tl_if_blank:nTF{ #3 }{\cdot\, \middle|\, }{#3\, \middle|\, }%
	}
	\IfValueT{#2}{#2}%
\right)%
}
\ExplSyntaxOff

\NewDocumentCommand {\cbrace} {t+ D[]{black} D(){\widthof{#5}} m m } {%
	\begingroup%
		\color{#2}
		\IfBooleanTF{#1}{%
			\overbrace{#4}^%
		}{
			\underbrace{#4}_%
		}%
		{\parbox[c]{#3}{\centering\footnotesize{#5}}}%
	\endgroup% 
}

\let\oldforall\forall
\renewcommand{\forall}{\oldforall \, }

\let\oldexist\exists
\renewcommand{\exists}{\oldexist \, }

\makeatletter

\newcommand{\rankcolor}[2]{%
	\expandafter\renewcommand\csname #1\endcsname[1]{%
		\ifblank{##1}{%
			{\color{#2} \textbf{#2}}%
		}{%
			\ifmmode
				\textcolor{#2}{\bm{##1}}%
			\else%
				{\color{#2} \textbf{##1}}%
			\fi	
		}%
	}
}

\rankcolor{first}{red}
\rankcolor{second}{blue}
\rankcolor{third}{cyan}
\makeatother

\graphicspath{{./Figures/}{./figures/}}
\DeclareDocumentCommand{\includeCroppedPdf}{ o O{./Figures/} m }{
	\IfFileExists{#2#3-crop.pdf}{}{%
		\immediate\write18{pdfcrop #2#3.pdf #2#3-crop.pdf}}%
	\includegraphics[#1]{#2#3-crop.pdf}
}

\makeatletter
\newcommand*{\addFileDependency}[1]{% argument=file name and extension
  \typeout{(#1)}
  \@addtofilelist{#1}
  \IfFileExists{#1}{}{\typeout{No file #1.}}
}
\makeatother

\definecolor{gray90}{gray}{0.9}
\def\colorlist{red,blue,brown,cyan,darkgray,gray,lightgray,green,lime,magenta,olive,orange,pink,purple,teal,violet,white,yellow}

\makeatletter
\def\startcomment{[}
\ifx\nohighlights\undefined
	\newcommand{\createcolor}[1]{%
			\expandafter\newcommand\csname #1\endcsname[1]{{\color{#1} ##1}}%
	}
	\newcommand{\msout}[1]{\text{\color{green} \sout{\ensuremath{#1}}}}
	\newcommand{\del}[1]{{\color{green}\ifmmode \msout{#1}\else\sout{#1}\fi}}
\else
	\newcommand{\createcolor}[1]{%
			\expandafter\newcommand\csname #1\endcsname[1]{%
				\noexpandarg%
				\StrChar{##1}{1}[\firstletter]%
				\if\firstletter\startcomment%
					\relax
				\else%
					##1
				\fi
			}%
	}
	\newcommand{\msout}[1]{}
	\newcommand{\del}[1]{}
\fi

\def\@tempa#1,{%
    \ifx\relax#1\relax\else
        \createcolor{#1}%
        \expandafter\@tempa
    \fi
}
\expandafter\@tempa\colorlist,\relax,
\makeatother

\newcommand{\hhide}[1]{}
\ifx\diagnoselabel\undefined
	\relax
\else
	\makeatletter
	\def\@testdef #1#2#3{%
		\def\reserved@a{#3}\expandafter \ifx \csname #1@#2\endcsname
			\reserved@a  \else
			\typeout{^^Jlabel #2 changed:^^J%
				\meaning\reserved@a^^J%
				\expandafter\meaning\csname #1@#2\endcsname^^J}%
			\@tempswatrue \fi}
	\makeatother
\fi

\crefname{question}{question}{questions}

\usepackage{tikz}
\usepackage{pgfplots}
\pgfplotsset{compat=1.5}%%%%%%%%%%%%%%%%%%%

\providecommand{\U}[1]{\protect\rule{.1in}{.1in}}

\theoremstyle{definition}

\newacronym{KRR}{KRR}{kernel ridge regression}
\newacronym{RBF}{RBF}{radial basis function}
\newacronym{RMSE}{RMSE}{relative mean-squared error}
\newacronym{GGSP}{GGSP}{generalized graph signal processing}
\newacronym{GSP}{GSP}{graph signal processing}
\newacronym{TSP}{TSP}{topological signal processing}
\newacronym{GFT}{GFT}{graph Fourier transform}

\begin{document}
%\ninept
%\thanks{This research is supported by the Singapore Ministry of Education Academic Research Fund Tier 2 grant MOE-T2EP20220-0002, and the National Research Foundation, Singapore and Infocomm Media Development Authority under its Future Communications Research and Development Programme.}
\title{A sampling construction of graphon 1-norm convergence}
\author{Xingchao Jian, Feng Ji, Wee Peng Tay \thanks{The authors are with School of Electrical and Electronic Engineering, Nanyang Technological University, Singapore.}}

%\address{School of Electrical and Electronic Engineering, Nanyang Technological University, Singapore}

\maketitle %\thispagestyle{empty}
%\ninept

\begin{abstract}
In the short note, we describe a sampling construction that yields a sequence of graphons converging to a prescribed limit graphon in $1$-norm. This convergence is stronger than the convergence in the cut norm, usually used to study graphon sequences. The note also contains errata of the previous version of the note.
\end{abstract}

%\begin{keywords}
%Graphon signal processing, cut norm, $1$-norm, operator norm, convergence
%\end{keywords}

\section{Introduction}
Recall that a graphon is a symmetric measurable function $\bW: [0,1]^2 \to [0,1]$ \cite{Lov:12} (see \cite{Lua21} for a signal processing framework based on the graphon theory). It is a device that can be used to study large graphs, as graphons are limit objects of graph sequences with the cut norm $\norm{\cdot}_{\square}$ defined by
\begin{align*}
  \norm{\bW}_{\square} = \sup_{U,V \subset [0,1]} |\int_{U,V} W(x,y) \ud x\ud y|,   
\end{align*}
where the supreme is taken over all pairs of measurable subsets $U,V$ of $[0,1]$. Apparently, for any $\bW$, we have $\norm{\bW}_{\square} \leq \norm{\bW}_1$, where $\norm{\cdot}_1$ is the usual $L^1$-norm. On the other hand, the cut norm is strictly weaker than  $L^1$-norm $\norm{\cdot}_1$. There are sequences of graphons converging in $\norm{\cdot}_{\square}$ that do not converge in $\norm{\cdot}_1$, e.g., the canonical graphons associated with a sequence of ER graphs sampled from a limit graphon $\bW$ \cite{Lov:12}.   

However, it is usually more difficult to deal with $\norm{\cdot}_{\square}$ as one needs to the supreme taken over $U,V \subset [0,1]$. In this short note, we give a sampling construction that generates a sequence of graphons that converges to a prescribed limit graphon in $\norm{\cdot}_1$, which is a strong notion of convergence.  

\section{The sampling construction} \label{sec:aeo}

In this section, we describe the promised sampling construction. Assume that $\bW$ is Riemann integrable. This is true if $\bW$ is continuous a.e., by the Lebesgue integrability criterion.

We define a sequence of random graphs of size $n$ ($n\in\bbN$):
\begin{itemize}
    \item Partition $[0,1]$ into $n$ equal intervals of size $1/n$ each.
    \item Uniformly random sample $x_i$ from the $i$-th interval.
    \item Form random graph $\bG_n$ on $v_1,\ldots, v_n$ such that there is an edge $(v_i,v_j)$ with probability $\bW(x_i,x_j)$ if $i\neq j$. 
\end{itemize}

Recall that, for a graph $\bG = (\calV,\calE)$, $\calV = \set{v_i\given i=1,\dots,n}$, its associated canonical graphon is the following step function 
\begin{align*}
    \bW_{\bG}(x,x')= \sum_{i,j=1}^n \bS_{ij} I_i(x)I_j(x'),
\end{align*}
where $\bS=(\bS_{ij})$ is the adjacency matrix of $\bG$, and $I_{i}$, for $i=1,2,\dots,n$, are the indicator functions of the sets $[0, \ofrac{n}), [\frac{1}{n}, \frac{2}{n}),\dots,[\frac{n-1}{n}, 1]$. 

For each $n$, let $\bW_{\bG_n}$ be the canonical graphon associated with $\bG_n$. It is a step function on the $n\times n$ grid of $[0,1]^2$. We have essentially obtained a distribution of graphs in the sense of \cite{Ji23}. If we view $T_{\bW_{\bG_n}}$ as a random operator, then convolution amounts to taking expectation as in \cite{Ji23}. Motivated by this observation, we shall see that the expected graphon converges in $1$-norm in this section.

To state a general result, we formally define $k$-th power of graphons. Given graphons $\bW$ and $\bW'$, their product $\bW\odot\bW'$ is the graphon such that $T_{\bW\odot\bW'} = T_{\bW}\circ T_{\bW'}$. Explicitly,  
\begin{align*}
    \bW\odot\bW'(x,y) = \int_{[0,1]} \bW(x,z)\bW'(z,y) \ud z.
\end{align*}
Denote the Lebesgue measure on $\mathbb{R}^n$ by $\mu_{\mathbb{R}^n}$. For convenience, we write $\mu$ for $\mu_{\mathbb{R}^2}$. Let $\bW^{\odot k}$ be the $k$-th fold product of $\bW$ (with itself). We first notice that $\bW^{\odot k}$ is Riemann integrable for any $k\geq 1$. 

\begin{Lemma} \label{lem:iwa}
    If $\bW$ and $\bW'$ are both Riemann integrable, then so is $\bW\odot\bW'$.
\end{Lemma}

The technical proof is contained in the appendix. The following is the result on $1$-norm convergence. 

\begin{Theorem}
If $\bW$ is Riemann integrable, then $(\mathbb{E}(\bW_{\bG_n}))^{\odot k} \to \bW^{\odot k}$ in $1$-norm as $n\to \infty$ for any $k\geq 1$. 
\end{Theorem}
\begin{proof}
   We first prove the case $k=1$. Notice that for each sample of the random graph $\bG_n$, its canonical graphon $\bW_{\bG_n}$ is a step function on $n\times n$ grid of $[0,1]^2$, so does $\mathbb{E}(\bW_{\bG_n})$. For $i\neq j$, on the $(i,j)$-th grid $I_{ij}$ of $[0,1]^2$, $\mathbb{E}(\bW_{\bG_n})$ takes the value 
    \begin{align*}
    w_{n,i,j} = \frac{1}{n^2}\int_{I_{ij}} \bW(x_i,x_j)\ud x_i\ud x_j,
    \end{align*}
    which is the average of $\bW$ on $I_{ij}$. Hence, it is bounded between the supreme and infimum of $\bW$ on $I_{ij}$, i.e., 
    \begin{align*}
    \inf_{i,j} \stackrel{\Delta}{=} \inf_{(x_i,x_j)\in I_{ij}}\bW(x_i,x_j) & \leq w_{n,i,j} \\ & \leq \sup_{i,j} \stackrel{\Delta}{=} \sup_{(x_i,x_j)\in I_{ij}}\bW(x_i,x_j).
    \end{align*}
    Observe that outside the diagonal blocks $\cup_{1\leq i\leq n} I_{ii}$, the integration of $\mathbb{E}(\bW_{\bG_n})$ is a Riemann sum. We can therefore estimate $\norm{\mathbb{E}(\bW_{\bG_n})- \bW}_{1}$ as: 
\begin{align*}
           & \norm{\mathbb{E}(\bW_{\bG_n})- \bW}_{1} \\
       \leq  & \sum_{1\leq i\neq j \leq n} \int_{(x,y) \in I_{ij}} |\bW(x,y)-w_{n,i,j}| \ud x \ud y \\
       & + \sum_{1\leq i\leq n}\int_{(x,y) \in I_{ii}} |\bW(x,y)| \ud x\ud y \\
       \leq & \sum_{1\leq i\neq j \leq n} \int_{(x,y) \in I_{ij}} |\bW(x,y)-w_{n,i,j}| \ud x \ud y  + \frac{1}{n} \norm{\bW}_{\infty} \\
       \leq & \sum_{1\leq i\neq j \leq n} \int_{(x,y) \in I_{ij}} \big(\sup_{i,j}-\inf_{i,j}\big) \ud x \ud y  + \frac{1}{n} \norm{\bW}_{\infty}\\
       \leq & \sum_{1\leq i, j \leq n} \int_{(x,y) \in I_{ij}} \big(\sup_{i,j}-\inf_{i,j}\big) \ud x \ud y  + \frac{1}{n} \norm{\bW}_{\infty}.
\end{align*}
The expression converges to $0$ as $n \to \infty$ as $\bW$ is assumed to be Riemann integrable. 

For general $k$, we first observe that by \cref{lem:iwa} and induction, $\bW^{\odot k}$ is Riemann integrable. Moreover, $(\mathbb{E}(\bW_{\bG_n}))^{\odot k}$ is a step function that is constant on each of the $n\times n$ grid of $[0,1]^2$. We may argue inductively that for each $(x,y)$ in the $(i,j)$-th grid $I_{ij}$ with $i\neq j$, we have the relations 
\begin{align*}
\inf_{(x',y')\in I_{ij}} \bW^{\odot k}(x',y') &\leq (\mathbb{E}(\bW_{\bG_n}))^{\odot k}(x,y) \\ &\leq \sup_{(x',y')\in I_{ij}} \bW^{\odot k}(x',y').
\end{align*}
By the same argument as $k=1$, we have $(\mathbb{E}(\bW_{\bG_n}))^{\odot k} \to \bW^{\odot k}, n\to \infty$ in $1$-norm.
\end{proof}

We first remark that the result for $k>1$ can also be alternatively proved by induction without using \cref{lem:iwa}. One needs to estimate integrals of the form 
\begin{align*}
\int_{(x,y)} |\int_z \bW'_1(x,z)\big(\bW_1(z,y)-\bW_2(z,y)\big)\ud z| \ud y\ud x
\end{align*}
when we split the difference 
\begin{align*}
& \bW_1'(x,z)\bW_1(z,y)- \bW_2'(x,z)\bW_2(z,y) \\
= & \bW'_1(x,z)\big(\bW_1(z,y)-\bW_2(z,y)\big) \\
& + \big(\bW'_1(x,z) - \bW_2'(x,z)\big) \bW_2(z,y).
\end{align*}

For this, we
use $|\bW'_1(x,z)| \leq 1$ and the Fubini theorem:
\begin{align*}
&\int_{(x,y)} |\int_z \bW'_1(x,z)\big(\bW_1(z,y)-\bW_2(z,y)\big)\ud z| \ud y\ud x\\
\leq & \int_{(x,y)} \int_z |\bW'_1(x,z)\big(\bW_1(z,y)-\bW_2(z,y)\big)|\ud z \ud y\ud x \\
\leq & \int_{y} \int_z |\bW_1(z,y)-\bW_2(z,y)|\ud z \ud y = \norm{\bW_1-\bW_2}_1.
\end{align*}
Details are omitted here.

We have $1$-norm convergence of a polynomial in $\mathbb{E}(\bW_{\bG_n})$ to that in $\bW$ for the same polynomial. Having $1$-norm convergence is always desirable, as it is the strongest version of convergence. However, the trade-off is that to empirically approximate $\mathbb{E}(\bW_{\bG_n})$, we need to draw sufficiently many samples. For each draw, a different $x_i$ might be chosen from the $i$-th interval.

\section{Errata of the previous version}
In the previous version of the note, we have mistakenly claimed that certain results in \cite{Lua21} are incorrect. After going through the reply \cite{Rui24} and the original paper, we agree that the statements of  \cite{Lua21} are correct and an error in a proof is fixable, which is done in \cite{Rui24}. We apologize to the authors of \cite{Lua21} regarding our wrong claims. 

\section{Conclusions}\label{sect:conc}

In the theory of graphon convergence, the $1$-norm convergence is strictly stronger than the convergence in the cut norm. However, it is usually easier to deal with the $1$-norm. We have described a sampling construction that yields a sequence of graphons converging to a prescribed limit graphon in $1$-norm, in contrast to known constructions in the literature that yield sequences converge in the cut norm. 

\appendix

\section{}
\begin{proof}[Proof of \cref{lem:iwa}]
By the Lebesgue integrability criterion, it suffices to show that $\bW''=\bW\odot\bW'$ is continuous a.e. Suppose on the contrary that $\bW''$ is discontinuous on a set $D$ with Lebesgue measure $\mu(D)>0$. This means there is $\epsilon>0$ such that for every $(x,y)\in D, \delta>0$, there is $(x',y')$ such that $\norm{(x,y)-(x',y')} < \delta$ and $|\bW''(x,y)-\bW''(x',y')|\geq \epsilon$. 

Let $S \subset [0,1]^2$ be the union of sets of discontinuous points of $\bW$ and $\bW'$. We have $\mu(S)=0$. For any $\epsilon_1>0$, there is an open subset $U$ containing $S$ and $\mu(U)\leq \epsilon_1$. We use $U\setcomp$ to denote the complement of $U$ in $[0,1]^2$. It is bounded and closed, and hence compact. As $\bW, \bW'$ are both continuous on $U^c$, they are both uniformly continuous on $U^c$. 

We claim that for any $\epsilon_0>0$, there is a $\delta_1$ such that for any $(x',y') \in U^c$ and $(x,y) \in [0,1]^2$ satisfying $\norm{(x,y)-(x',y')} < \delta_1$, we have 
\begin{align*}
|\bW(x,y)-\bW(x',y')|< \epsilon_0, |\bW'(x,y)-\bW'(x',y')|< \epsilon_0.
\end{align*}

To see this, we first notice that as $U^c$ is compact, we have a $\delta_0$ (in place of $\delta_1$) for the conclusion on $\bW,\bW'$ when both $(x,y), (x',y') \in U^c$. For each $z= (x',y') \in U^c$, there is $r_z$ such that for any $w$ in $B_z(r_z)$, the open ball centered at $z$ with radius $r_z$, we have the desired conclusion for $\bW,\bW'$ for $\epsilon_0/4$. Hence the same holds for any two points in $B_z(r_z)$ by triangle inequality with $\epsilon_0/2$. However, $U^c$ is compact, and we have a finite subcover $\{B_i,1\leq i\leq m\}$ of $\{B_z(r_z), z\in U^c\}$. For each $B_i$ in the subcover, let $d_i$ be its minimum distance to the union of $B_j \cap U^c$ for $B_j$ not intersecting $B_i$. We have $d_i>0$ as $U^c$ is closed and we only need to consider finitely many of $B_j$ (in the subcover). Notice that the boundary $\partial$ of the open set $\cup_{1\leq i \leq m}B_i$ is compact and disjoint from $U^c \subset \cup_{1\leq i\leq m}B_i$. There is a positive minimum distance $d_0$ between $\partial$ and $U_c$. We set $\delta_1 = \min \{d_i,1\leq i\leq m\}\cup \{\delta_0\}\cup \{d_0\}$. If $\norm{(x,y)-(x',y')} < \delta_1$ with $(x',y')\in U^c\cap B_i$ for some $B_i$, then either (a) $(x,y) \in U^c$ or (b) $(x,y)$ belongs to some $B_j$ (as $\delta_1 \leq d_0$) such that $B_i\cap B_j \neq \emptyset$ (as $\delta_1 \leq d_i$). The choice of $\delta_0$ settles the former case. For the latter case, comparing with a point $(x'',y'') \in B_i\cap B_j$, we have the desired conclusion by triangle inequality. 

Now for any $\epsilon_2>0$, let $S_{\epsilon_2}$ be the subset of $[0,1]$ consisting of $x$ such that $\mu_{\mathbb{R}^1}(x \times [0,1]\cap U)\geq \epsilon_2$. We have $\mu_{\mathbb{R}^1}(S_{\epsilon_2})\leq \epsilon_1/\epsilon_2$. Similarly, if $S_{\epsilon_2}'$ is the subset of $[0,1]$ consisting of $y$ such that $\mu_{\mathbb{R}^1}([0,1] \times y\cap U)\geq \epsilon_2$. We also have $\mu_{\mathbb{R}^1}(S_{\epsilon_2}')\leq \epsilon_1/\epsilon_2$.  Let $U' = U \cup (S_{\epsilon_2}\times [0,1]) \cup ([0,1] \times S_{\epsilon_2}')$. Its measure satisfies $\mu(U') \leq \epsilon_1+2\epsilon_1/\epsilon_2$. 

For $(x,y) \in U'^c \subset U^c$ and any $(x',y')$ such that $\norm{(x,y)-(x',y')} < \delta_1$, we estimate 
\begin{align*}
    &|\bW''(x,y) - \bW''(x',y')| = |\bW\odot\bW'(x,y)-\bW\odot\bW'(x',y')| \\
    = & |\int_{[0,1]} \bW(x,z)\bW'(z,y)- \bW(x',z)\bW'(z,y')\ud z| \\
    \leq & \int_{[0,1]} |\bW(x,z)\bW'(z,y)- \bW(x',z)\bW'(z,y')|\ud z \\
    \leq & \int_{[0,1]} |\bW(x,z)| \cdot |\bW'(z,y)- \bW'(z,y')|\ud z  \\
    & + \int_{[0,1]} |\bW'(z,y')|\cdot |\bW(x,z)- \bW(x',z)|\ud z \\
    \leq & \int_{[0,1]} |\bW'(z,y)- \bW'(z,y')|\ud z\\ & + \int_{[0,1]} |\bW(x,z)- \bW(x',z)|\ud z.
\end{align*}
For the first summand, notice that $|y-y'| < \delta_1$ and $y$ is outside $S_{\epsilon_2}'$. Therefore, for $z$ outside a set $S_y$ of measure $\epsilon_2$, we have $|\bW'(z,y)-\bW'(z,y')| < \epsilon_0$. Therefore, we estimate 
\begin{align*}
    &\int_{[0,1]} |\bW'(z,y)- \bW'(z,y')|\ud z \\ 
    = & \int_{S_y} |\bW'(z,y)- \bW'(z,y')|\ud z + \int_{S_y^c} |\bW'(z,y)- \bW'(z,y')|\ud z \\
    \leq & 2\mu_{\mathbb{R}^1}(S_y) + (1-\mu_{\mathbb{R}^1}(S_y))\epsilon_0 \\
    \leq & 2\epsilon_2 + \epsilon_0. 
\end{align*}
The same estimation applied to $\int_{[0,1]} |\bW(x,z)- \bW(x',z)|\ud z$ implies that $|\bW''(x,y) - \bW''(x',y')| \leq 4\epsilon_2 + 2\epsilon_0$. 

In summary, for every $(x,y) \in U'^c$ with $\mu(U') \leq \epsilon_1 + 2\epsilon_1/\epsilon_2$, there is $\delta_1$ such that: if $(x',y')$ satisfy $\norm{(x,y)-(x',y')} < \delta_1$, then $|\bW''(x,y)-\bW''(x',y')| < 4\epsilon_2+2\epsilon_0$. As $\epsilon_0,\epsilon_1, \epsilon_2$ can be chosen arbitrarily, we may set $\epsilon_2 = \sqrt{\epsilon_1}$ and $\epsilon_0, \epsilon_1$ small enough such that 
\begin{itemize}
    \item $\epsilon_1 + 2\epsilon_1/\epsilon_2  = \epsilon_1 + 2\sqrt{\epsilon_1} < \mu(D)$. %\blue{[What is $\epsilon_3$?]}
    \item $4\epsilon_2+2\epsilon_0 = 4\sqrt{\epsilon_1}+2\epsilon_0 <\epsilon$. 
\end{itemize}
Letting $\delta = \delta_1$ in the statement of the first paragraph, we obtain a contradiction. Hence, $\bW''$ is continuous and a.e., and hence it is Riemann integrable. 
\end{proof}

\bibliographystyle{IEEEtran}
\bibliography{bib/IEEEabrv,bib/StringDefinitions,bib/refs}

\end{document}